\documentclass[notitlepage,onecolumn]{quantumarticle}
    \pdfoutput=1
    \usepackage{tikz}
    \usepackage[sort&compress,numbers,merge]{natbib}

    \usepackage[strict]{revquantum}
        \newaffil{MSRQuArC}{
            Quantum Architectures and Computation Group,
            Microsoft Research,
            Redmond, WA, United States
        }
        \newoperator{AGF}
        \newoperator{Uni}
        \newoperator{Var}
        \newrm{spec} % for the spectral norm
        \newrm{true}

        \theoremstyle{definition}
        \newtheorem{definition}[theorem]{Definition}
        \newtheorem{example}{Example}

    \usepackage{xcolor}
        \colorlet{comment}{cud-bluish-green!30!black!70}

    \usepackage{algorithm}
    \usepackage{algpseudocode}
        \renewcommand{\Comment}{\hfill{\footnotesize/\!\!/} }
        \algblockdefx[Arguments]%
            {Arguments}{EndArguments}%
            {%\vskip0.5em%
                \color{comment}
                {\footnotesize \(\blacksquare\) \emph{Arguments}} \normalsize%
            }%
            {\color{black}}
        \algsetblock[Arguments]{Arguments}{EndArguments}{}{0.325cm}
        \renewcommand{\inlinecomment}[1]{{\color{comment} \Comment {\footnotesize #1} \normalsize}}
        
        \newcommand{\seccomment}[1]{%
            \vskip0.5em%
            \State {\color{comment} \footnotesize \(\blacksquare\) \emph{#1}} \normalsize
        }

        \algnotext{EndFor}
        \algnotext{EndIf}
        \algnotext{EndWhile}
        \algnotext{EndFunction}

    \usepackage[bold]{hhtensor}

    \newcommand{\figurefolder}{../fig}

    \newcommand{\pperp}{\perp\!\!\!\!\perp}
    \newcommand{\rref}{\mathrm{ref}}

    \DeclareSymbolFont{extraup}{U}{zavm}{m}{n}
    \DeclareMathSymbol{\varheart}{\mathalpha}{extraup}{86}
    \newcommand{\happy}{F}

    \newcommand{\lin}{\mathrm{L}}
    \newcommand{\unitary}{\mathrm{U}}
    \newoperator{Herm}
    \newcommand{\chan}{\mathrm{C}}
    \newcommand{\dens}{\mathrm{D}}
    \newcommand{\hil}{\mathcal{H}}
    \newcommand{\CC}{\mathbb{C}}
    \newcommand{\Conv}{\mathrm{Conv}}
    \newcommand{\LConst}{\mathcal{L}}

\renewcommand{\figurefolder}{fig}
\begin{document}

\title{Bayesian ACRONYM Tuning}
\date{authors in alphabetical order}

\author{John Gamble}
    \affilMSRQuArC

\author{Chris Granade}
    \affilMSRQuArC

\author{Nathan Wiebe}
    \affilMSRQuArC

\begin{abstract}
We provide an algorithm that uses Bayesian randomized benchmarking in concert with a local optimizer, such as SPSA, to find a set of controls that optimizes that average gate fidelity.
We call this method Bayesian ACRONYM tuning as a reference to the analogous ACRONYM tuning algorithm.  
Bayesian ACRONYM distinguishes itself in its ability to retain prior information from experiments that use nearby control parameters; whereas traditional ACRONYM tuning does not use such information and can require many more measurements as a result.
We prove that such information reuse is possible under the relatively weak assumption that the true model parameters are Lipshitz-continuous functions of the control parameters.  
We also perform numerical experiments that demonstrate that over-rotation errors in single qubit gates can be automatically tuned from $88\%$ to $99.95\%$ average gate fidelity using less than $1kB$ of data and fewer than $20$ steps of the optimizer.
\end{abstract}

\maketitle

%=============================================================================
\section{Introduction}
%=============================================================================

Tuning gates in quantum computers is a task of fundamental importance to building a quantum computer.
Without tuning, most quantum computers would have insufficient accuracy to implement a simple algorithm let alone achieve the stringent requirements on gate fidelity imposed by quantum error correction~\cite{fowler2009high,cross_comparative_2007}.  
Historically, qubit tuning has largely been done by experimentalists refining an intelligent initial guess for the physical parameters by hand to account for the ideosyncracies of the device.
Recently, alternatives have been invented that allow devices to be tuned in order to improve performance on real-world estimates of gate quality.
These methods, often based on optimizing quantities such as average gate fidelities, are powerful but come with two drawbacks.
At present all such methods require substantial input data to compute the average gate fidelity and estimate its gradient, and at present no method can use information from the history of an optimization procedure to reduce such data needs.
Our approach, which we call Bayesian ACRONYM tuning (or BACRONYM), addresses these problems.

BACRONYM is based strongly on the ACRONYM protocol invented by Ferrie and Moussa~\cite{fm_robust_2015}.
There are two parts to the ACRONYM gate tuning protocol.
The first uses randomized benchmarking \cite{magesan_characterizing_2012} to obtain an estimate of gate fidelity as a function of the controls.
The second optimizes the average gate fidelity using a local optimizer such as Nelder-Mead or stochastic gradient descent.
While many methods can be used to estimate the average gate fidelity, randomized benchmarking is of particular significance because of its ability to give an efficient estimate of the average gate fidelity under reasonable assumptions~\cite{pry+_what_2017}, and because of its amenability to experimental application \cite{heeres_implementing_2016}.
The algorithm then uses a protocol, similar to SPSA~\cite{spa_multivariate_1992}, to optimize the estimate of the gate fidelity by changing the experimental controls and continues to update the parameters until the desired tolerance is reached.

The optimization used in ACRONYM simply involves varying a parameter slightly and applying the fidelity estimation protocol from scratch every time.  
When the a quantum system is evaluated at two nearby points in parameter space, an operation performed repeatedly in descent algorithms, the objective function does not typically change much in practice.
Since ACRONYM does not take this into account, it requires more data than is strictly needed.
Thus, if ACRONYM could be modified to use prior information extracted from the previous iteration in SPSA, the data needed to obtain an estimate of the gradient can be reduced. 

Bayesian methods provide a natural means to use prior information within parameter estimation and have been used previously to analyze randomized benchmarking experiments.
These methods, yield estimates of the average gate fidelity based on prior beliefs of about the randomized benchmarking parameters as well as the evidence obtained experimentally~\cite{gfc_accelerated_2015}.   
To use a Bayesian approach, we begin by taking as input a probability distribution for the average gate fidelity ($\AGF$) as function of the control parameters $\vec{\theta}$,  $\Pr(\AGF|\vec{\theta})$. 
This is our prior belief about the average gate fidelity.
In addition to a prior, we need a method for computing the likelihood of witnessing a set of experimental evidence $E$. 
This is known as the likelihood function; in the case of Bayesian randomized benchmarking, it is $\Pr(E|\AGF;\vec{\theta})$.
Given these as input, we then seek to output an approximation to the posterior probability distribution, \emph{i.e.,} the probability with which the AGF takes a specific value conditioned on our prior belief and $E$. 
To accomplish this, we use Bayes' theorem, which states that
\begin{equation}
    \Pr(\AGF|E;\vec{\theta}) = \frac{\Pr(\AGF|\theta) \Pr(E|\AGF;\theta)}{\Pr(E|\theta)},\label{eq:Bayes}
\end{equation}
where $\Pr(E|\theta)$ is just a normalization constant.
From the posterior distribution $\Pr(\AGF|E;\vec{\theta})$ we can then extract a point estimate of the $\AGF$  (by taking the mean) or estimate its uncertainty (by computing the variance).

Our work combines these two ideas to show that provided the quantum channels that describe the underlying gates are continuous functions of the control parameters then the uncertainty in parameters like~$\AGF$ that occurs from transitioning from $\vec{\theta}\rightarrow \vec{\theta'}$ in the optimization process is also a continuous function of $\|\vec{\theta} - \vec{\theta'}\|$.
This gives us a rule that we can follow to argue how much uncertainty we have to add to our posterior distribution $\Pr(\AGF|E;\vec{\theta})$ to use it as a prior $\Pr(\AGF|\vec{\theta'})$ at the next step of the gradient optimization procedure.

%-----------------------------------------------------------------------------
\subsection{Notation}
%-----------------------------------------------------------------------------

The notation that we use in this paper necessarily spans several fields, most notably Bayesian inference and randomized benchmarking theory.
Here we will introduce the necessary notation from these fields in order to understand our results.
For any distribution $\Pr(\vec{x})$ over a vector $\vec{x}$ of random variables, we write $\supp(\Pr(\vec{x}))$ to mean the set of vectors $\vec{x}$ such that $\Pr(\vec{x})>0$.
When it is clear from context, we will write $\supp(\vec{x} | \vec{y})$ in place of $\supp(\Pr(\vec{x} | \vec{y}))$.

Let $\hil = \CC^{d}$ be a finite-dimensional Hilbert space describing the states of a quantum system of interest, and let $\lin(\hil)$ be the set of linear operators acting on $\hil$.
Let $\Herm(\hil) \subsetneq \lin(\hil)$ and $\unitary(\hil) \subsetneq \lin(\hil)$ be the sets of Hermitian and unitary operators acting on $\hil$, respectively.
For the most part, however, we are not concerned directly with pure states $\ket{\psi} \in \hil$, but with classical distributions over such states, described by density operators $\rho \in \dens(\hil) \subsetneq \Herm(\hil) \subsetneq \lin(\hil)$.
Whereas $\hil$ transforms under $\unitary(\hil)$ by left action, $\dens(\hil)$ transforms under $\unitary(\hil)$ by the \emph{group action} $\bullet : \unitary(\hil) \times \lin(\hil) \to \lin(\hil)$, given by $U \bullet \rho \defeq U\rho U^\dagger$.
We note that $\bullet$ is linear in its second argument, such that for a particular $U \in \unitary(\hil)$, $U \bullet : \lin(\hil) \to \lin(\hil)$ is a linear function.
We thus write that $U \bullet {} \in \lin(\lin(\hil))$.
Moreover, since $U \bullet {}$ is a completely positive and trace preserving map on $\lin(\hil)$, we say that $U \bullet {}$ is a \emph{channel} on $\hil$, written $\chan(\hil) \subsetneq \lin(\lin(\hil)) \subsetneq \lin(\hil) \to \lin(\hil)$.
More generally, we take $\chan(\hil)$ to be the set of all such completely positive and trace preserving maps acting on $\lin(\hil)$.

%-----------------------------------------------------------------------------
\subsection{Problem Description}
%-----------------------------------------------------------------------------

Before proceeding further, it is helpful to carefully define the problem that we address with BACRONYM.
In particular, let $G = \langle V_0, \dots, V_{\ell - 1} \rangle \subsetneq \unitary(\hil)$ be a group and a unitary 2-design \cite{dankert_exact_2006}, such that $G$ is appropriate for use in standard randomized benchmarking.
Often, $G$ will be the Clifford group acting on a Hilbert space of dimension $d$, but smaller twirling groups may be chosen in some circumstances \cite{ian_pc}.
We will consider that the generator $T$ is a gate, which we would like to tune to be $V_0$ without loss of generality, as a function of a vector $\vec{\theta}$ of control parameters, such that $T = T(\vec{\theta})$.
We write that $V_i \pperp \vec{\theta}$ for all $i \ge 0$ to indicate that the generators $\{V_0, \dots, V_{\ell - 1}\}$ are not functions of the controls $\vec{\theta}$ (note that $V_0$ is manifestly not a function of the controls because it represents the ideal action).
Nonetheless, it is often convenient to write that $V_i = V_i(\vec{\theta)})$ with the understanding that $\partial_{\theta_j} V_i = 0$ for all $i \ge 0$ and for all control parameters $\theta_j$.

In order to reason about the errors in our implementation of each generator, we will write that the imperfect implementation $\tilde{V} \in \chan(\hil)$ of a generator $V \in \{V_0, \dots, V_{\ell - 1}\}$ is defined as
\begin{align}
    \tilde{V} & = \Lambda_V (V \bullet {}) \label{eq:Lambda_V_def} \\
\intertext{which acts on $\rho$ as }
    \tilde{V}[\rho] & = \Lambda_V[V\rho V^{\dagger}],
\end{align}
where $\Lambda_V$ is the \emph{discrepancy channel} describing the errors in $V$.
Note that for an ideal implementation, $\Lambda_V$ is the identity channel.

We extend this definition to arbitrary elements of $G$ in a straightforward fashion.
Let $U \defeq \prod_{i \in \vec{i}(U)} V_{i}$, where $\vec{i}(U)$ is the sequence of indices of each generator in the decomposition of $U$.
For instance, if $G = \langle H, S \rangle$ for the phase gate $S = \diag(1, \ii)$, then $\sqrt{X} = HSH$ is represented by $\vec{i}(U) = (0, 1, 0)$.
Combining the definition of $U$ with Eq.~(\ref{eq:Lambda_V_def}), the imperfect composite action $\tilde{U}$ is
\begin{align}
    \tilde{U} &= \prod_{i \in \vec{i}(U)} \tilde{V_i}
	 	 = \prod_{i \in \vec{i}(U)} \Lambda_{V_i} (V_i \bullet {})
		 \defeq \Lambda_U (U \bullet),
\end{align}
where the final point defines the composite discrepancy channel $\Lambda_U$.
By rearranging the equation above, we obtain
\begin{align}
    \label{eq:discrepancy-u}
    \Lambda_U = \tilde{U} (U^\dagger \bullet {}) =
        \left(\prod_{i \in \vec{i}(U)} \Lambda_{V_i} (V_i \bullet {}) \right)
        \left(U^\dagger \bullet {} \right).
\end{align}
Returning to the example $\sqrt{X} = HSH$, we thus obtain that
\begin{align}
    \Lambda_{\sqrt{X}} = \Lambda_H (H \bullet {}) \Lambda_S (S \bullet {}) \Lambda_H (H \bullet {})
                         ((H^\dagger S^\dagger H^\dagger) \bullet {})
\end{align}
is the discrepancy channel describing the noise incurred if we implement $\widetilde{\sqrt{X}}$ as the sequence $\tilde{H}\tilde{S}\tilde{H}$.

Equipped with the discrepancy channels for all elements of $G$, we can now concretely state the parameters of interest to randomized benchmarking over $G$.
Standard randomized benchmarking without sequence reuse \cite{gfc_accelerated_2015}, in the limit of long sequences \cite{wal_randomized_2017}, depends only on the state preparation and measurement (SPAM) procedure and on the average gate fidelity $\AGF(\Lambda_{\rref})$, where
\begin{align}
    \label{eq:ref-channel-defn}
    \Lambda_{\rref} \defeq \expect_{U \sim \Uni(G)} [\Lambda_U] = \frac{1}{|G|} \sum_{U \in G} \Lambda_U
\end{align}
is the reference discrepancy channel, obtained by taking the expectation value of the discrepancy channel $\Lambda_U$ over $U$ sampled uniformly at random from $G$, and where the average gate fidelity is given by the expected action of a channel $\Lambda$ over the Haar measure $\dd\psi$,
\begin{align}
    \AGF(\Lambda) \defeq \int \dd{\psi} \braket{
        \psi \mid
        \Lambda(\ket{\psi}\bra{\psi})
        \mid \psi
    }.
\end{align}
When discussing the quality of a particular generator, say $T\defeq \tilde{V_0}$, we unfortunately cannot directly access  $\AGF(\Lambda_{T})$ experimentally.
However, interleaved randomized benchmarking allows us to rigorously estimate $\AGF(\Lambda_{T} \Lambda_{\rref})$ in the limit of long sequences and without sequence reuse.

Our goal here is to find a set of control parameters that optimizes $\AGF(\Lambda_{T} \Lambda_{\rref})$. 
To state this more formally, suppose that $T$ is a function of a vector $\vec{\theta}$ of control parameters such that $T = T(\vec{\theta})$.
For all ideal generators, we write that $V_i \pperp \vec{\theta}$ for all $i \ge 0$ to indicate that the other generators $\{V_0, \dots, V_{\ell - 1}\}$ are not functions of the controls $\vec{\theta}$.
We also assume that $\tilde V_i \pperp \vec{\theta}$ for all $i > 0$, so that $T(\vec{\theta}) = \Lambda_{V_0} (\vec \theta ) V_0$ is the sole generator we are optimizing.
We therefore aim to find $\vec{\theta}$ such that $\vec{\theta} = {\rm argmax} \left(\AGF(\Lambda_{T(\vec{\theta})} \Lambda_{\rref})\right)$.

This problem has previously been considered by \citet{ew_adaptive_2014} and later by~\citet{kelly2014optimal}, who proposed the use of interleaved randomized benchmarking with least-squares fitting to implement an approximate oracle for $\AGF(\Lambda_T(\vec{\theta}) \Lambda_{\rref}(\vec{\theta}))$.
Taken together with the bounds showed by \citet{mgj+_efficient_2012}~and later improved by \citet{kdr+_robust_2014}, this approximate oracle provides an approximate lower bound on $\AGF(\Lambda_{\rref}(\vec{\theta}))$.
This lower bound can then be taken as an objective function for standard optimization routines such as Nelder--Mead to yield a ``fix-up'' procedure that improves gates based on experimental evidence.
\citet{fm_robust_2015} showed an improvement in this procedure by the use of an optimization algorithm that is more robust to the approximations incurred by the use of finite data in the underlying randomized benchmarking experiments.
In particular, the simultaneous pertubative stochastic approximation (SPSA) \cite{spa_multivariate_1992}, while less efficient for optimizing exact oracles, can provide dramatic improvements in approximate cases such as that considered by \citet{fm_robust_2015}.
This advantage has been further shown in other areas of quantum information, such as in tomography \cite{fer_self_2014,cfa_experimental_2016}.

We improve this result still further by using a Lipschitz continuity assumption on the dependence of $\Lambda_T$ on $\vec{\theta}$ to propagate prior information between optimization iterations.
This assumption is physically well-motivated: it reflects a desire that our control knobs have a smooth (but not known) influence on our generators.
Since small gradient steps cannot greatly modify the average gate fidelity of interest under such a continuity assumption, the prior distribution for each randomized benchmarking experiment is closely related to the posterior distribution from the previous optimization iteration.

Recent work has shown, however, that this approach faces two significant challenges.
First, the work of \citet{pry+_what_2017} has shown explicit counterexamples in which reconstructing $\AGF(\Lambda_{T}(\vec{\theta}))$ from $\AGF(\Lambda_{T}(\vec{\theta}) \Lambda_{\rref}(\vec{\theta}))$ can yield very poor estimates due to the gauge dependence of this inverse problem.
Second, the work of \citet{hwf+_bayesian_2018} has shown that the statistical inference problem induced by randomized benchmarking becomes considerably more complicated with sequence reuse, and in particular, depends on higher moments such as the unitarity \cite{wghf_estimating_2015}.
While the work of \citet{hwf+_bayesian_2018}~provides the first concrete algorithm that allows for learning randomized benchmarking parameters with sequence reuse, we will consider the single-shot limit to address the \citet{pry+_what_2017} argument, as this is the unique randomized benchmarking protocol that provides gauge invariant estimates of $\AGF(\Lambda_{T}(\vec{\theta}) \Lambda_{\rref}(\vec{\theta}))$ \cite{rpz_gauge_2017}, and as this model readily generalizes to include the effects of error correction \cite{cgff_logical_2017}.

In this work, we adopt as our objective function
\begin{align}
    \happy(\vec{\theta}) \defeq \AGF(\Lambda_{T}(\vec{\theta}) \Lambda_{\rref}(\vec{\theta})).
\end{align}
This choice of objective represents that we want to see improvements in the interleaved average gate fidelity, regardless of whether they occur from a more accurate target gate or a more accurate reference channel.
In practice, these two contributions to our objective function can be teased apart by the use of more complete protocols such as gateset tomography~\cite{merkel2013self,blume2017demonstration}.
We proceed in three steps.
First, we demonstrate that the Lipschitz continuity of $\Lambda_T(\vec{\theta})$ implies the Lipschitz continuity of $\happy(\vec{\theta})$.
We then proceed to show that this implies an upper bound on $\Var[\happy(\vec{\theta} + \vec{\delta\theta}) | \text{data}]$ in terms of $\Var[\happy(\vec{\theta}) | \text{data}]$, such that we can readily produce estimates $\hat{\happy}(\vec{\theta})$ at each step of an optimization procedure, while reusing much of our data to accelerate the process.
Finally, we conclude by presenting a numerical example for a representative model to demonstrate how BACRONYM may be used in practice.

%=============================================================================
\section{Lipschitz Continuity of $\happy(\vec{\theta})$}
%=============================================================================
Proving Lipshitz continuity of the objective function is an important first step towards arguing that we can reuse information during BACRONYM's optimization process.
We need this fact because if the objective function were to vary unpredictably at adjacent values of the controls then finding the optima would reduce to an unstructured search problem, which cannot be solved efficiently.
Our aim is to first argue that continuity of $\Lambda$ implies continuity of $\happy$.  
We then will use this fact to argue about the maximum amount that the posterior variance can grow as the control parameters are updated, which will allow us to quantify how to propagate uncertainties of $\happy$ at adjacent points later.
We begin by recalling the definition of Lipschitz continuity for functions acting on vectors.

\begin{definition}[Lipschitz continuity]
    \label{def:lipschitz-functions}
    Given a Euclidean metric space $S$, a function $f : S \to \mathbb{R}$ is said to be Lipschitz continuous if there exists $\LConst \ge 0$ such that for all $\vec{x}, \vec{y} \in S$,
    \begin{align}
        |f(\vec{x}) - f(\vec{y})| \le \LConst \|\vec{x} - \vec{y}\|.
    \end{align}
    If not otherwise stated, we will assume $\| \cdot \|$ on vectors to be the Euclidean norm $\| \cdot \|_2$.
\end{definition}

As an example, $f(x) = \sqrt{x}$ is not Lipschitz continuous on [0,1], but any differentiable function on a closed, bounded interval of the real line is.
We now generalize the notion of Lipschitz continuity to channels.
Let $\lin(\hil)$ be the set of all linear operators acting on the Hilbert space $\hil$, and let $\lin(\lin(\hil))$ be the set of linear operators acting on all such linear operators (often referred to as \emph{superoperators}).

\begin{definition}[Lipschitz continuity of channels]
    \label{def:lipschitz-channels}
    Given a metric space $S$ and a Hilbert space $\hil$, we say that a function $\Lambda : S \to \lin(\lin(\hil))$ is $\LConst$-continuous or Lipschitz continuous in the $\star$ distance if there exists $\LConst \ge 0$ such that for all $\vec{x}, \vec{y} \in S$ and $\rho \in \dens(\hil)$,
    \begin{align}
        \| \Lambda(\vec{x})[\rho] - \Lambda(\vec{y})[\rho] \|_\star \le \LConst \|\vec{x} - \vec{y}\|.
    \end{align}
    If not specified explicitly, the trace norm $\| \cdot \| = \| \cdot \|_{\Tr}$ is assumed for operators in $\lin(\hil)$.
\end{definition}

From the definition, we immediately can show the following:

\begin{lemma}[Composition of Lipschitz continuous channels]
    \label{lem:compose-chan-lip}
    Let $\Lambda, \Phi : S \to \lin(\lin(\hil))$ be Lipschitz continuous in the trace distance with constants $\LConst$ and $\mathcal{M}$, respectively.
    Then, $(\Phi \Lambda) : \vec{x} \mapsto \Phi(\vec{x}) \Lambda(\vec{x})$ is Lipschitz continuous in the trace distance with constant $\LConst + \mathcal{M}$.
\end{lemma}
\begin{proof}
 The proof of the lemma follows immediately after a few applications of the triangle inequality under the assumption of continuity of the individual channels.
    \begin{align*}
        \| (\Phi \Lambda)(\vec{x})[\rho] - (\Phi \Lambda)(\vec{y})[\rho] \|_{\Tr} & =
            \| \Phi(\vec{x}) [\Lambda(\vec{x}) [\rho]] - \Phi(\vec{y}) [\Lambda(\vec{y})  \rho]] \|_{\Tr} \\
        & =
            \| \Phi(\vec{x}) [\Lambda(\vec{x}) [\rho]] - \Phi(\vec{x}) [\Lambda(\vec{y}) [\rho]] + \Phi(\vec{x}) [\Lambda(\vec{y}) [\rho]] - \Phi(\vec{y}) [\Lambda(\vec{y}) [\rho]] \|_{\Tr} \\
        & \le
            \| \Phi(\vec{x}) [\Lambda(\vec{x}) [\rho]] - \Phi(\vec{x}) [\Lambda(\vec{y}) [\rho]] \|_{\Tr} +
            \| \Phi(\vec{x}) [\Lambda(\vec{y}) [\rho]] - \Phi(\vec{y}) [\Lambda(\vec{y}) [\rho]] \|_{\Tr} \\
        & \le
            \| \Phi(\vec{x}) [\Lambda(\vec{x}) [\rho]] - \Phi(\vec{x}) [\Lambda(\vec{y}) [\rho]] \|_{\Tr} +
            \mathcal{M} \|\vec{x} - \vec{y}\| \\
        & \le
            \| \Lambda(\vec{x}) [\rho] - \Lambda(\vec{y}) [\rho] \|_{\Tr} +
            \mathcal{M} \|\vec{x} - \vec{y}\| \\
        & \le
            \LConst \|\vec{x} - \vec{y}\| +
            \mathcal{M} \|\vec{x} - \vec{y}\|,
        \end{align*}
    where the second-to-last line follows from contradiction on Helstrom's theorem \cite{wat_theory_2018}.
\end{proof}

We note that the above lemma immediately implies that if $\Lambda(\vec{\theta})$ is Lipschitz continuous in the trace distance with constant $\LConst$, then so is $(\Phi \Lambda)(\vec{\theta})$ for any channel $\Phi \pperp \vec{\theta}$, since $\Phi$ can be written as a channel that is Lipschitz continuous in the trace distance with constant $0$.

\begin{corollary}[Composition of multiple Lipschitz continuous functions and channels]
\label{cor:comp-mult-channels}
Let $\Lambda_0, \Lambda_1,...,\Lambda_k : S \to \lin(\lin(\hil))$ be Lipschitz continuous in the trace distance with constants $\LConst_i$ with $i \in [0,1,...,k]$. Then, $(\Lambda_0 \Lambda_1 \cdots \Lambda_k) : \vec{x} \mapsto \Lambda_0(\vec{x}) \Lambda_1(\vec{x})\cdots \Lambda_k(\vec{x}) $ is Lipschitz continuous in the trace distance with constant $\sum_{i = 0}^k \LConst_i$.
\end{corollary}

\begin{lemma}
    \label{lem:convex-chan-lip}
    Let $\Lambda : S \to \lin(\lin(\hil))$ be a convex combination of channels,
    \begin{align}
        \Lambda(\vec{\theta}) & = \sum_i p_i \Lambda_i(\vec{\theta}),
    \end{align}
    where $\{p_i\}$ are nonnegative real numbers such that $\sum_i p_i = 1$, and where each $\Lambda_i : S \to \lin(\lin(\hil))$ is Lipschitz continuous in a norm $\|\cdot\|_\star$ with constant $\LConst_i$.
    Then, $\Lambda$ is Lipschitz continuous with constant $\bar{\LConst} = \sum_i p_i \LConst_i$.
\end{lemma}
\begin{proof}
    Consider an input state $\rho \in \dens(\hil)$.
    Then,
    \begin{align*}
        \| \Lambda(\vec{\theta})[\rho] - \Lambda(\vec{\theta}')[\rho] \|_\star & =
            \left\|
                \sum_i p_i \left(
                    \Lambda_i(\vec{\theta})[\rho] - \Lambda_i(\vec{\theta}')[\rho]
                \right)
            \right\|_\star \\
        & \le
            \sum_i p_i \left(
                \left\|
                        \Lambda_i(\vec{\theta})[\rho] - \Lambda_i(\vec{\theta}')[\rho]
                \right\|_\star
            \right) \\
        & \le
            \sum_i p_i \LConst_i \| \vec{\theta} - \vec{\theta}' \| \\
        & = \bar{\LConst} \| \vec{\theta} - \vec{\theta}' \|.
    \end{align*}
\end{proof}

The above lemmas can then be used to show that $\AGF(\Lambda_T(\vec{\theta}))$ is Lipschitz continuous with constant $\LConst$ when $\Lambda_T(\vec{\theta})$ is Lipschitz continuous in the trace distance with constant $\LConst$, as we formally state in the following theorem.

\begin{theorem}
    \label{thm:agf-continuity}
    Let $\Lambda(\vec{\theta})$ be Lipschitz continuous in the trace distance with constant $\LConst$.
    Then $\AGF(\Lambda(\vec{\theta}))$ is Lipschitz continuous with constant $\LConst$.
\end{theorem}
\begin{proof}
    Recall that
    \begin{align}
        \AGF(\Lambda(\vec{\theta})) &
            \defeq \int \dd\psi \braket{
                \psi | \Lambda\left[ \ket{\psi}\bra{\psi} \right] | \psi
            }, \\
        \intertext{so}
        |\AGF(\Lambda(\vec{\theta})) - \AGF(\Lambda(\vec{\theta}'))| & =
            \left|
                \int \dd\psi \braket{
                    \psi |
                        \Lambda(\vec{\theta})  [\ket{\psi} \bra{\psi}] -
                        \Lambda(\vec{\theta}') [\ket{\psi} \bra{\psi}]
                    | \psi
                }
            \right| \\
        & \le
            \int \dd\psi \left|
                \braket{
                    \psi |
                        \Lambda(\vec{\theta})  [\ket{\psi} \bra{\psi}] -
                        \Lambda(\vec{\theta}') [\ket{\psi} \bra{\psi}]
                    | \psi
                }
            \right| \\
        & \le
            \int \dd\psi \left\|
                \Lambda(\vec{\theta})  [\ket{\psi} \bra{\psi}] -
                \Lambda(\vec{\theta}') [\ket{\psi} \bra{\psi}]
            \right\|_{\Tr} \\
        & \le
            \int \dd\psi\,\LConst\|\vec{\theta} - \vec{\theta}'\| \\
        & = \LConst\|\vec{\theta} - \vec{\theta}'\|.
    \end{align}
\end{proof}

As noted in the introduction, we do not have direct access to $\AGF(\Lambda_T(\vec{\theta}))$, but rather to $\AGF(\Lambda_T(\vec{\theta}) \Lambda_{\rref}(\vec{\theta}))$.
In particular, $\happy(\vec{\theta}) \defeq \AGF(\Lambda_T(\vec{\theta}) \Lambda_{\rref}(\vec{\theta}))$ may be estimated from the interleaved randomized benchmarking parameters:
\begin{subequations}
    \label{eq:rb-parameter-defns}
    \begin{align}
        p(\vec{\theta}) & \defeq \frac{d\happy(\vec{\theta}) - 1}{d - 1}, \label{eq:pDef}\\
        A(\vec{\theta}) & \defeq \Tr(E \Lambda_{\rref}(\vec{\theta})[\rho - \frac{\id}{d}]), \\
        \text{and }
        B(\vec{\theta}) & \defeq \Tr(E \Lambda_{\rref}(\vec{\theta})[\frac{\id}{d}]),
    \end{align}
\end{subequations}
where $d = \operatorname{dim}(\hil)$, $\rho$ is the state prepared at the start of each sequence, and $E$ is the measurement at the end of each sequence.
We consider $A$ and $B$ later, but note for now that up to a factor of $d / (d - 1)$, Lipschitz continuity of $\happy(\vec{\theta})$ immediately implies Lipschitz continuity of $p(\vec{\theta})$.
Thus, we can follow the same argument as above, but using the channel $\Lambda_T (\vec{\theta}) \Lambda_{\rref}(\vec{\theta})$ instead to argue the Lipschitz continuity of experimentally accessible estimates.

\begin{table}
    \begin{center}
        \begin{tabular}{rl}
        \hline
        $G_{0}$ & $\{\openone, H\}$ \\
    $G_{1}$ & $\{S, HS, SH, HSH\}$ \\
    $G_{2}$ & $\{SS, HSS, SHS, SSH, HSHS, HSSH\}$ \\
    $G_{3}$ & $\{SSS, HSSS, SHSS, SSHS, HSHSS, HSSHS, SHSSH, HSHSSH\}$ \\
    $G_{4}$ & $\{SHSSS, SSHSS, HSHSSS, HSSHSS\}$ \\
        \hline
    \end{tabular}
    \end{center}
    \caption{
        \label{tab:example-partition}
        A partitioning of the twirling group $G = \langle H, S\rangle$ based on the number of occurrences of the target gate $T = S$ in the expansion of each element.
    }
\end{table}

We proceed to show the Lipschitz continuity of $\happy$ and hence of $p$ by revisiting the definition \autoref{eq:ref-channel-defn} of $\Lambda_{\rref}$.
In particular, we partition the twirling group as $G = \bigcup_{n = 0}^{\infty} G_n$, where $G_n$ is the set of elements of $G$ whose decomposition into generators $\{T, V_1, \dots, V_{\ell - 1}\}$ requires at least $n$ instances of the target gate $T$.
For instance, if $G = \langle S, H \rangle$ and the target gate is $T = S$, then $Z \in G_2$ since $Z = SS$ is the decomposition of $Z$ requiring the least copies of $S$.
The partition of $G$ in this example is shown as \autoref{tab:example-partition}.

Using this partitioning of $G$, we can define an analogous partition on the terms occuring in the definition of $\Lambda_{\rref}(\vec{\theta})$,
\begin{align}
    \Lambda_{\rref}(\vec{\theta}) & = \sum_{n = 0}^{\infty} \frac{| G_n |}{| G |} \Lambda_{\rref, n}(\vec{\theta}), \\
    \text{where }
    \label{eq:rref-n-expansion}
    \Lambda_{\rref, n}(\vec{\theta}) & \defeq \frac{1}{| G_n |} \sum_{U \in G_n} \Lambda_{U}(\vec{\theta}).
\end{align}

\begin{theorem}
    \label{thm:ref-continuity}
    If $\Lambda_T(\vec{\theta})$ is Lipschitz continuous in the trace distance with constant $\LConst$, then $\Lambda_{\rref,n}(\vec{\theta})$ is Lipschitz continuous in the trace distance with constant $nL$.  Furthermore $\Lambda_{\rref}(\vec{\theta})$ is Lipshitz continuous with constant $\bar{n} \defeq \sum_{n = 0}^{\infty} n \frac{|G_n|}{|G|}$.
\end{theorem}
\begin{proof}
Consider one of the summands from \autoref{eq:rref-n-expansion}, and without loss of generality let $U = V_{i_0} V_{i_1} \cdots V_{i_k}$ for the sequence of integer indices $\vec{i} = (i_0, i_1, \dots, i_k)$. 
Then, by \autoref{eq:discrepancy-u},
\begin{align}
    \Lambda_U(\vec{\theta}) & =
        \Lambda_{V_{i_0}}(\vec{\theta})
        (V_{i_0} \bullet {})
        \cdots
        \Lambda_{V_{i_k}}(\vec{\theta})
        (V_{i_k} \bullet {})
        (U^\dagger \bullet {}).
\end{align}
Note that, $\forall i$, $V_i \pperp \vec \theta$ since these are ideal channels and hence independent of the control vector $\vec \theta$; these channels are Lipschitz continuous in the trace distance with constant $0$.
Further, each $\Lambda_{V_{i}} \pperp \vec \theta$ for $i>0$; these channels are also Lipschitz continuous in the trace distance with constant $0$.
By assumption, we have $\Lambda_{V_{0}}$ is Lipschitz continuous in the trace distance with constant $\LConst$.
Hence, each factor in $\Lambda_U$ is Lipschitz continuous in the trace distance with constant $\LConst$ or $0$, as detailed above.

By \autoref{cor:comp-mult-channels}, $\Lambda_U$ is Lipschitz continuous in the trace distance with constant $mL$, where $m$ counts the number of $0$s in $\vec i$ (corresponding to the number of times the target gate occurs in the decomposition of $U$).
By construction, $m\leq n$, so $\Lambda_U$ is also Lipschitz continuous in the trace distance with constant $nL$.

Using \autoref{lem:convex-chan-lip} to, we now have that $\Lambda_{\rref,n}(\vec{\theta})$ is Lipschitz continuous in the trace distance with constant $\frac{1}{| G_n |} \sum_{U \in G_n} nL = nL$, which is what we wanted to show.

We thus have that $\Lambda_{\rref}(\vec{\theta})$ is Lipschitz continuous in the trace distance with constant $\bar{n}\LConst$, wherein
\begin{align}
    \bar{n} \defeq \sum_{n = 0}^{\infty} n \frac{|G_n|}{|G|}
\end{align}
is the average number of times that the target gate $T$ appears in decompositions of elements of the twirling group $G$.
\end{proof}

\begin{corollary}[Lipschitz continuity of $\Lambda_{\rref}(\vec{\theta})$]
\label{cor:cont-ref-channel}
Let 
\begin{align}
    \bar{n} \defeq \sum_{n = 0}^{\infty} n \frac{|G_n|}{|G|}
\end{align}
be the average number of times that the target gate $V_0$ appears in decompositions of elements of the twirling group $G$.
Then, $\Lambda_{\rref}(\vec{\theta})$ is Lipschitz continuous in the trace distance with constant $\bar{n}\LConst$.
\end{corollary}

Combining with the previous argument, we thus have our central theorem.

\begin{theorem}
    \label{thm:pab-continuity}
    Let $\Lambda_T(\vec{\theta})$ be Lipschitz continuous in the trace distance with constant $\LConst$.
    Then, $\happy(\vec{\theta})=\AGF(\theta)$ is Lipschitz continuous with constant $(1 + \bar{n}) \LConst$, and $p(\vec{\theta})$ is Lipschitz continuous with constant $d (1 + \bar{n}) \LConst/ (d - 1) $, and       $A(\vec{\theta})$ and $B(\vec{\theta})$ are Lipschitz continuous with constant $\bar{n} \LConst$.
\end{theorem}
\begin{proof}
	
    First, $\happy(\vec{\theta})=\AGF(\Lambda_T (\vec{\theta}) \Lambda_{\rref}(\vec{\theta}))$.
    By assumption, $\Lambda_T(\vec{\theta})$ is Lipschitz continuous with constant $\LConst$, and by \autoref{cor:cont-ref-channel}, $\Lambda_{\rref}(\vec{\theta}))$ is Lipschitz continuous with constant $\bar{n} \LConst$.
    Hence, by \autoref{thm:agf-continuity}, $\happy(\vec{\theta})$ is Lipschitz continuous with constant $(1 + \bar{n}) \LConst$.

    Next, recall that $p(\vec{\theta}) =  \frac{d \happy(\vec{\theta})- 1}{d - 1}$. 
    Then, it follows that $p(\vec{\theta})$ is Lipschitz continuous with constant $\frac{d (1+\bar n) \LConst}{d - 1}$.
   
   For $B(\vec{\theta})$, we have
    \begin{align}
        \label{eq:pab-continuity-step0}
        | B(\vec{\theta}') - B(\vec{\theta}) | & =
            \left|
                \Tr(E \Lambda_{\rref}(\vec{\theta}')[\id / d]) -
                \Tr(E \Lambda_{\rref}(\vec{\theta}) [\id / d])
            \right| \\
        & = \left| \Tr(E (\Lambda_{\rref}(\vec{\theta}') - \Lambda_{\rref}(\vec{\theta})) [\id / d])\right|.
    \end{align}
    Letting $(\epsilon_0,\epsilon_1,...,\epsilon_d)$ be the ordered singular values of $E$ and $(\lambda_0,\lambda_1,...,\lambda_d)$ be the ordered singular values of $(\Lambda_{\rref}(\vec{\theta}') - \Lambda_{\rref}(\vec{\theta})) [\id / d]$, we have 
    \begin{align}
        \label{eq:pab-continuity-step1}
        | B(\vec{\theta}') - B(\vec{\theta}) | 
        \leq \sum_{i=1}^{d} \epsilon_i \lambda_i 
        \leq \max (\epsilon) \sum_{i=1}^{d}\lambda_i 
        = \max(\epsilon) \|(\Lambda_{\rref}(\vec{\theta}') - \Lambda_{\rref}(\vec{\theta})) [\id / d]\|_{\Tr} 
        \leq \bar{n} \LConst,
    \end{align}
    Since $E$ and $C$ are both Hermitian, $EC$ is also Hermitian, and thus $\|EC\|_{\Tr} = \Tr(|EC|) \ge |\Tr(EC)|$.
    The argument is completed by H\"older's inequality~\cite{wat_theory_2018}, which states that for all $X$ and $Y$, $\|XY\|_{\Tr} \le \|X\|_{\Tr} \|Y\|_{\spec}$, where $\|\cdot\|_{\spec}$ is the spectral norm (\emph{a.k.a.} the induced $(2 \to 2)$-norm or Schatten $\infty$-norm).
    In particular, we note that since $E$ is a POVM effect, $\|E\|_{\spec} \le 1$, such that $\|EC\|_{\Tr} \le \|C\|_{\Tr} \le 1$.

Finally, we note that this argument goes identically for the state $\rho - \frac{\id}{d}$, as we did not use any special properties of $\frac{\id}{d}$.
Hence, we also have that $| A(\vec{\theta}') - A(\vec{\theta}) | \leq\bar{n} \LConst$.
\end{proof}

We are thusly equipped to return to the problem of estimating $\happy(\vec{\theta} + \vec{\delta\theta})$ from experimental data concerning $\happy(\vec{\theta})$.

\begin{theorem}\label{thm:lipshitz}
    Suppose that $f(\vec{\theta}, \vec{y})$ is a Lipschitz continuous function of $\vec{\theta}$ with constant $\LConst$ where $y$ is a variable in a measurable set $S$ with corresponding probability distribution on that set of $\Pr(\vec{y})$ and for any function $g:S\mapsto \mathbb{R}$ define $\mathbb{E}_{\vec{y}}(g(\vec{y})) = \int_{S} g(\vec{y}) \Pr(\vec{y})\dd\vec{y} $ and $\Var_{\vec{y}}(g(\vec{y})) = \mathbb{E}_{\vec{y}} \big(g(\vec{y}) - \mathbb{E}_{\vec{y}}(g(\vec{y}) \big)^2$.
    For all $\vec{\theta}$ and $\vec{\theta}'$ such that $\LConst \|\vec{\theta}' - \vec{\theta}\| < \sqrt{\Var_{\vec{y}}(f(\vec{\theta}, \vec{y})})$, it holds that
    \begin{align}
        \Var_{\vec{y}}[f(\vec{\theta}', \vec{y})]
            & \le
                \Var_{\vec{y}}[f(\vec{\theta}, \vec{y})]
                \left(
                    1 + \frac{
                        2 \LConst \|\vec{\theta}' - \vec{\theta}\|
                    }{
                        \sqrt{\Var_{\vec{y}}[f(\vec{\theta}, \vec{y})]}
                    }
                \right).
    \end{align}
\end{theorem}
\begin{proof}
    Note that since $f$ is Lipschitz continuous as a function of $\vec{\theta}$,
    \begin{align}
    \left| f(\vec{\theta}', \vec{y}) - f(\vec{\theta}, \vec{y}) \right| \leq \LConst \|\vec{\theta}' - \vec{\theta}\|,
    \end{align}
    so there exists a function $c$ such that $|c(\vec{\theta}, \vec{\theta}', \vec{y})| \le 1$ for all $\vec{\theta}$, $\vec{\theta}'$ and $\vec{y}$:
    \begin{align}
        f(\vec{\theta}', \vec{y})
            & =
                f(\vec{\theta}, \vec{y}) +
                \LConst \|\vec{\theta}' - \vec{\theta}\| c(\vec{\theta}, \vec{\theta}', \vec{y}).
    \end{align}
    Thus, $\Var_{\vec{y}}[c] \le 1$, and by addition of variance, we have that
    \begin{align}
        \Var_{\vec{y}}[f(\vec{\theta}', \vec{y})]
                        & = \Var_{\vec{y}}[f(\vec{\theta}, \vec{y})]+\LConst^2 \|\vec{\theta} - \vec{\theta}'\|^2 \Var_{\vec{y}}(c(\vec{\theta}, \vec{\theta}', \vec{y}))\! +\! 2\LConst \|\vec{\theta} - \vec{\theta}'\| \text{Cov}_{\vec{y}}(f(\theta,\vec{y}),c(\theta,\theta',\vec{y}))\nonumber\\
            & \le
                \Var_{\vec{y}}[f(\vec{\theta}, \vec{y})] +
                \LConst^2 \|\vec{\theta} - \vec{\theta}'\|^2 +
                \LConst \|\vec{\theta} - \vec{\theta}'\| \sqrt{\Var_{\vec{y}}(f(\vec{\theta}, \vec{y})}.\nonumber\\
                & \le
                \Var_{\vec{y}}[f(\vec{\theta}, \vec{y})] +
               2 \LConst \|\vec{\theta} - \vec{\theta}'\| \sqrt{\Var_{\vec{y}}(f(\vec{\theta}, \vec{y})}.
    \end{align}
    The result then follows from elementary algebra.
\end{proof}

%-----------------------------------------------------------------------------
\subsection{Examples}
%-----------------------------------------------------------------------------

\begin{figure}
    \begin{center}
        \includegraphics[width=0.95\linewidth]{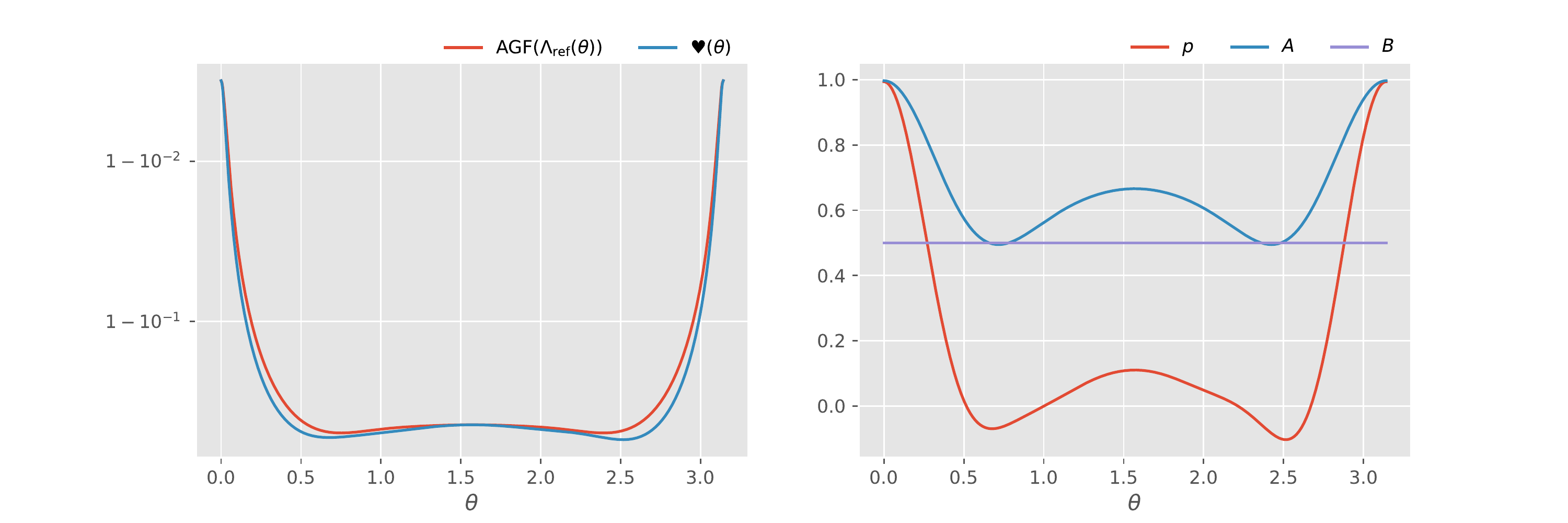}
    \end{center}
    \caption{
        \label{fig:unitary-overrotation-objective}
        The objective function $\happy(\theta)$ and the average gate fidelity versus the overrotation angle $\theta$ for \autoref{ex:unitary-overrotation} is given in the left figure.  The right figure gives the calculated RB parameters as a function of $\theta$ where the optimal solution $\theta=0$ is unknown to the optimizer a priori.
    }
\end{figure}

\begin{example}[Lipschitz Continuity of Unitary Overrotation]
    \label{ex:unitary-overrotation}
    % Include generated results.
    
        % Last generated: 2018-06-13T07:29:01.688612
        \newcommand{\unitaryOverrotationObjectiveLipschitz}{1.48}
        \newcommand{\unitaryOverrotationRefAGFLipschitz}{1.13}

    Consider $G = \langle S, H \rangle$, where $T = S$ is the target gate.
    For a control parameter vector consisting of a single overrotation parameter $\vec{\theta} = (\delta\theta)$, suppose that $\Lambda_T[\rho] = (\e^{-\ii\,\delta\theta\,\sigma_z}) \bullet \rho$.
    Since this is a unitary channel, its Choi--Jami\l{}kowski rank\footnote{Sometimes informally called a ``Kraus rank.''} is 1.
    Thus, the AGF of $\Lambda_T$ can be calculated as the trace \cite{nie_simple_2002,hhh_general_1999,emerson_scalable_2005}
    \begin{align}
        \AGF(\Lambda_T(\delta \theta)) = \frac{ |\Tr(e^{-\ii\,\delta\theta\,\sigma_z})|^2 + 2 }{4 + 2}
            = \frac23 + \frac13 \cos(2\,\delta\theta).
    \end{align}
    On the other hand, $\happy(\delta\theta)$ isn't as straightforward, and so we will consider its Lipschitz continuity instead.
    To do so, we note that for all $\rho \in \dens(\mathbb{C}^2)$, we wish to bound the trace norm
    \begin{align}
        \Delta & = \| \Lambda_T(\delta\theta)[\rho] - \Lambda_T(\delta\theta')[\rho] \|_{\Tr}.
    \intertext{
        Expanding $\rho$ in the unnormalized Pauli basis as $\rho = \id / 2 + \vec{r} \cdot \vec{\sigma} / 2$, we note that since $\Lambda_T(\delta\theta)[\id] = \id$ and $\Lambda_T(\delta\theta)[\sigma_z] = \sigma_z$ for all $\delta\theta$, the above becomes 
    }
        \Delta & =
            \frac12 \| \Lambda_T(\delta\theta)[r_x \sigma_x + r_y \sigma_y +r_z \sigma_z] - \Lambda_T(\delta\theta')[r_x \sigma_x + r_y \sigma_y+r_z \sigma_z] \|_{\Tr} \\
&=            \frac12 \| \Lambda_T(\delta\theta)[r_x \sigma_x + r_y \sigma_y] - \Lambda_T(\delta\theta')[r_x \sigma_x + r_y \sigma_y] \|_{\Tr} \\
        & =
            4 |\sin(\delta\theta - \delta\theta')| \sqrt{r_x^2 + r_y^2} \\
        & \le
            4 |\sin(\delta\theta - \delta\theta')| \\
        & \le 4 |\delta\theta - \delta\theta'|,
    \end{align}
    where the last line follows from that $|\sin(x)| \le |x|$.
    Thus, we conclude that $\Lambda_T$ is Lipschitz continuous in the trace distance with constant 4.

    We can then find $\bar{n}$ for occurrences of $T$ in decompositions of elements of $G$ to find the Lipschitz constant for $\happy(\delta\theta)$ in this example.
    In particular, as shown in the Supplementary Material, $\bar{n} = 13 / 6$ for the presentation of the Clifford group under consideration, such that $\happy$ is Lipschitz continuous with constant $(d / (d - 1)) \times 4 \times (19 / 6) = 76 / 3$ in this case.

    We note that a more detailed analysis of the Lipschitz continuity of $\Lambda_T$ or a presentation of $G$ that is less dense in $T$ would both yield smaller Lipschitz constants for $\happy$, and hence better reuse of prior information.
    Thus by \autoref{thm:pab-continuity}, a change in overrotation of approximately $1 / 100$ the current standard deviation in $\happy$ would result in at most a doubling of the current standard deviation.

   We can easily include the effects of noise in other generators in numerical simulations.
    In particular, suppose that $\Lambda_H$ is a depolarizing channel with strength $0.5\%$.
    Then, simulating $\happy(\vec{\theta})$ for this case shows that $\happy$ is Lipschitz continuous with a constant of approximately \unitaryOverrotationObjectiveLipschitz, as illustrated in \autoref{fig:unitary-overrotation-objective}.
\end{example}

%=============================================================================
\section{Approximate Bayesian Inference}
%=============================================================================

An important implication of \autoref{thm:lipshitz} is that the uncertainty quantified by the variance of the posterior distribution yielded by Bayesian inference grows by at most a constant factor.
However, while the theorem specify how the variance should grow in the worst case scenario it does not give us an understanding of what form the posterior distribution should take.
Our goal in this section is to provide an operationally meaningful way to think about how the posterior distribution evaluated at $\vec{\theta}$ changes as the control parameters transition to $\vec{\theta'}$.

Let the posterior probability distribution for the objective function $\happy$ evaluated at parameters $\vec{\theta}$ be $\Pr\left(\happy(\vec{\theta})\right)$.
In practice, we do not generally estimate the objective function $\happy$ directly, but estimate $\happy$ from a latent variable $\vec{y}$, such as the RB parameters \autoref{eq:rb-parameter-defns}.
Marginalizing over this latent variable, we obtain the Bayesian mean estimator for $\happy$,

\begin{equation}
    \label{eq:marginalized-bme-happy}
    \hat{\happy} =
        \int \happy \Pr\left(\happy | \theta \right) \dd\happy
    =
        \int \happy \Pr\left(\happy | \theta, \vec{y} \right) \Pr(\vec{y}) \dd\vec{y}.
\end{equation}
For the RB case in particular, the objective function $\happy$ does not depend on the control parameters $\vec{\theta}$ if we know the RB parameters $\vec{y}$ exactly.
That is, we write that $\happy \pperp \vec{\theta} | \vec{y}$ for the RB case, such that $\Pr(\happy | \vec{\theta}, \vec{y}) = \Pr(\happy | \vec{y})$.
Moreover, $\Pr(\happy | \vec{y})$ is a $\delta$-distribution supported only at $\happy = (dp + 1) / (d + 1)$ where $\vec{y} = (p, A, B)$.
We may thus abuse notation slightly and write that $\happy = \happy(\vec{y})$ is a deterministic function.
Doing so, our estimator simplifies considerably, such that
\begin{equation}
    \label{eq:simplified-bme-happy}
    \hat{\happy} =
        \int \happy \Pr\left(\happy | \theta, \vec{y} \right) \Pr(\vec{y}) \dd\vec{y}
    =
        \int \happy(\vec{y}) \Pr(\vec{y}) \dd \vec{y}.
\end{equation}

In exact Bayesian inference, the probability density $\Pr(\vec{y})$ is an arbitrary distribution, but computation of the estimator \autoref{eq:simplified-bme-happy} is in general intractable.
Perhaps the most easily generalizable distribution is the sequential Monte Carlo (SMC) approximation \cite{dj_tutorial_2011}, also known as a particle filter, which attempts to approximate the probability density as
\begin{equation}
    \Pr\left(\happy | \theta, \vec{y} \right) \Pr(\vec{y} | \vec{\theta}) =
        \Pr(\happy, \vec{y} | \vec{\theta})
    \approx
        \sum_{j=1}^{N_p} w_j \delta(\vec{y} - \vec{y}_j) \delta(\happy_i - \happy),
\end{equation}
where $\delta$ is the Dirac-delta distribution and $\sum_j w_j =1$.
This representation is convenient for recording on a computer, as it only needs to store $(w_i, \vec{y}_i, \happy_i)$ for each particle.
If $\happy = \happy(\vec{y})$ is a deterministic function of the RB parameters then we need not even record $\happy$ with each particle, such that
\begin{equation}
    \Pr\left(\happy | \theta, \vec{y} \right) \Pr(\vec{y} | \vec{\theta}) \approx
        \sum_{j=1}^{N_p} w_j \delta(\vec{y} - \vec{y}_j) \delta(\happy(\vec{y}) - \happy).
\end{equation}

More generally, the SMC approximation allows us to approximate expectation values over the probability distribution using a finite number of points, or particles, such that the expectation value of any continuous function can be approximated with arbitrary accuracy as $N_p \rightarrow \infty$.
In particular, we can approximate the estimator $\hat{\happy}$ within arbitrary accuracy.

The uncertainty (mean squared error) of this estimator is given by the posterior variance,
\begin{equation}
    \mathbb{V}(\happy) =
        \int \happy^2 \Pr(\happy | \theta, \vec{y}) \Pr(\vec{y}) \dd\vec{y} -
        \hat{\happy}^2.
\end{equation}
The posterior variance can be computed as the variance over the variable $\vec{y}$ induced from the sequential Monte Carlo approximation to the probability distribution,
\begin{equation}
    \mathbb{V}(\happy) \approx
        \sum_i w_i \happy(\vec{y}_i)^2 - \left(\sum_i w_i \happy(\vec{y}_i)\right)^2,
\end{equation}
where we have assumed that $\happy \pperp \vec{\theta} | \vec{y}$ and that $\Pr(\happy | \vec{\theta})$ is a $\delta$-distribution, as in the RB case.
This observation is key to our implementation of Bayesian ACRONYM tuning.

A final note regarding approximate Bayesian inference is that the learning process can be easily implemented.
From~\autoref{eq:Bayes} if $\Pr(\happy|\theta, \vec{y}) \Pr(\vec{y}) = \sum_{j=1}^{N_p} w_j \delta(\vec{y}-\vec{y}_j)$ and if evidence $E$ is obtained in an experiment, then Bayes' theorem when applied to the weights $w_j$ yields
\begin{equation}
    w_j \gets \frac{\Pr(E | \vec{y_j}) w_j}{\sum_j \Pr(E | \vec{y_j}) w_j}.
\end{equation}
This update procedure is repeated iteratively over all data that is collected from a set of experiments.
In practice, if an accurate estimate is needed then an enormous number of particles may be needed because the weights shrink exponentially with the number of updates.
This causes the effective number of particles in the approximation to shrink exponentially and with it the accuracy of the approximation to the posterior.
We can address this by moving the particles to regions of high probability density.
In practice, we use a method proposed by \citet{lw_combined_2001} to move the particles but other methods exist and we recommend reviewing~\cite{dj_tutorial_2011,granade2017structured, hincks2018bayesian} for more details.
Here, we will use the implementation of particle filtering and Liu--West resampling provided by the QInfer package \cite{granade2017qinfer}.

%-------------------------------------------------------------------------------
\subsection{Reusing Priors from Nearby Experiments}
%-------------------------------------------------------------------------------

We have argued above that the posterior variance of the probability distribution is Lipshitz continuous, which allows us to reason that the variance of the probability distribution at most expands by a fixed multiplicative constant when transitioning information between different points.
Operationally though, it is less clear how we should choose the posterior distribution over the average gate fidelity in Bayesian ACRONYM training given prior information at a single point.
\autoref{thm:pab-continuity} provides us with an intuition that can be used for this: each element in the support of the probability distribution is shifted by at most a fixed amount that is dictated by the Lipshitz constants for the channels.
Here, we build on this intuition by showing that the prior at each step in a Bayesian ACRONYM tuning protocol can be related to the previous step in terms of the Minkowski sum and convex hull.

\begin{definition}[Convex hull]
    \label{def:convex-hull}
    Let $A$ be a set of vectors.
    Then the convex hull of $A$, written $\Conv(A)$ is the smallest convex set containing $A$,
    \begin{align}
        \Conv(A) \defeq \left\{
            \lambda \vec{a} + (1 - \lambda) \vec{b} :
            \vec{a}, \vec{b} \in A,
            0 \le \lambda \le 1
        \right\}.
    \end{align}
\end{definition}

\begin{definition}[Minkowski sum]
    \label{def:minkowski-sum}
    Let $A$ and $B$ be sets of vectors.
    Then the Minkowski sum $A + B$ is defined as the convolution of $A$ with $B$,
    \begin{align}
        A + B \defeq \left\{
            \vec{a} + \vec{b} : \vec{a} \in A, \vec{b} \in B
        \right\}.
    \end{align}
\end{definition}

With these concepts in place we can now state the following Corollary, which can be used to define a sensible prior distribution for $\vec{y}{(\vec{\theta} + \vec{\delta\theta})}$ given a posterior distribution for $\vec{y}(\vec{\theta})$.

\begin{corollary}
    \label{cor:prior-reuse}
    Let $\Lambda_T(\vec{\theta})$ be Lipshitz continuous in the trace distance with constant $\LConst$, and let $\Pr(\vec{y} | \vec{\theta})$ be a probability distribution over the RB parameters $\vec{y} = (p, A, B)$ for $\Lambda_T$ evaluated at some particular $\vec{\theta}$.
    Then, for any $\vec{\delta\theta} \in \mathbb{R}^n$, let
    \begin{align}
        \Delta & \defeq
            \|\vec{\delta\theta}\|, \\
        D & \defeq
            \Biggr\{ \pm \Delta \frac{d \LConst (1 + \bar{n})}{d - 1} \Biggr\} \times
            \Biggr\{ \pm \Delta (1 + \bar{n}) \LConst \Biggr\} \times
            \Biggr\{ \pm \Delta (1 + \bar{n}) \LConst \Biggr\}, \\
        \textrm{and }
        \Pr(\vec{y} | \vec{\theta} + \vec{\delta \theta}) & \defeq
            \frac{1}{8}\sum_{\vec{s} \in S} \Pr(\vec{y} - \vec{s} | \vec{\theta}).
    \end{align}
    The following statements then hold:
    \begin{enumerate}
        \item $\Pr(\vec{y} | \vec{\theta} + \vec{\delta\theta})$ is a valid prior probability distribution for $\vec{y}(\vec{\theta} + \vec{\delta\theta})$.
        \item $\hat{y} = \int \vec{y} \Pr(\vec{y} | \vec{\theta}) \dd\vec{y} = \int \vec{y} \Pr(\vec{y} | \vec{\theta} + \vec{\delta\theta}) \dd\vec{y}$.
        \item If $\Pr(\vec{y} | \vec{\theta})$ has support only on $A \subset \mathbb{R}^3$, then $\Pr(\vec{y}|\vec{\theta} + \vec{\delta\theta})$ has support only on $\Conv(A + D)$.
        \item If $\vec{y}_{\true}(\theta) \in A$ then $\vec{y}_{\true}(\theta+\delta\theta)\in \Conv(A + D)$.
    \end{enumerate}
\end{corollary}
\begin{proof}
    The proof of the first claim is trivial and follows immediately from the fact that $\Pr(\vec{y} | \vec{\theta})$ is a probability distribution.
    The proof of the second claim is also straightforward.
    Note that
    \begin{align}
        \hat{\vec{y}}
            \defeq \int \vec{y}
                \Pr(\vec{y} | \vec{\theta} + \vec{\delta\theta}) \dd\vec{y}
            = & \frac18 \int \sum_{\vec{s} \in \{\vec{y}\} + D}
                \vec{y} \Pr(\vec{y} - \vec{s} | \vec{\theta}) \dd\vec{y} \nonumber \\
            = & \frac18 \int \sum_{\vec{s} \in \{\vec{y}\} + D}
                (\vec{y} + \vec{s}) \Pr(\vec{y} | \vec{\theta}) \dd\vec{y} \nonumber \\
            = & \int \vec{y}
                \Pr(\vec{y} | \vec{\theta}) \dd\vec{y}.
    \end{align}

    To consider the third claim, let $\vec{c} = (c_p, c_A, c_B)$ be a vector such that $|c_p| \le dL(1+\bar{n})/(d-1)$ and $\max\{ |c_A|, |c_B|\} \le \LConst(1 + \bar{n})$.
    The convex hull $\Conv(D)$ consists of a convex region of identical dimensions.
    Since the set is convex it then follows that $\vec{c} \in \Conv(D)$.

    Put differently, we can express \autoref{def:lipschitz-functions} and \autoref{def:lipschitz-channels} in terms of the Minkowski sum, such that
    \begin{align}
        \vec{y}(\Lambda_T(\vec{\theta} + \vec{\delta\theta})) \in
            \Conv\left(
                \{\vec{y}(\Lambda_T(\vec{\theta})\} +
                D
            \right).
    \end{align}
    Taking the union over all vectors $\vec{a}$ in the support of $\Pr(\vec{y} | \vec{\theta})$, we obtain that
    \begin{align}
        \supp(\vec{y} | \vec{\theta} + \vec{\delta\theta}) \subseteq
            \Conv\left(
                \supp(\vec{y} | \vec{\theta}) +
                D
            \right).
    \end{align}
    From the linearity of convex hulls under Minkowski summation,
    \begin{equation}
        \Conv(\supp(\vec{y} | \vec{\theta}) + D) = \Conv(\supp(\vec{y} | \vec{\theta})) + \Conv(D).\label{eq:samething}
    \end{equation}
    The fourth and final statement then immediately follows from~\autoref{eq:samething}.
\end{proof}

This shows that if we follow the above rule to generate a prior distribution for the RB parameters at $\vec{\theta} + \vec{\delta\theta}$ then the resultant distribution does not introduce any bias into the current estimate of the parameters, which is codified by the mean of the posterior distribution.
We also have that if the true model is within the support of the prior distribution at $\vec{\theta}$ then it also will be at $\vec{\theta} + \vec{\delta\theta}$.
This is important because it states that we can use the resulting distribution to give a credible region for the RB parameters.
Thus this choice of prior is well justified and furthermore if the measurement process reduces the posterior variance faster than it expands when $\vec{\theta}$ is updated, it will allow us to get very accurate estimates of the true RB parameters without needing to extract redundant information.

%=============================================================================
\section{Numerical Experiments}
%=============================================================================
The above analysis shows that, under assumptions of Lipshitz continuity of the likelihood function, the posterior distribution found at a given step of the algorithm can be used to provide a prior for the next step.  This holds provided that we form a new prior that expands the variance of the posterior distribution.  

While the above analysis shows that prior information can be reused in theory, we will now show in practice that this ability to re-use prior information can reduce the information needed to calibrate a simulated quantum device.
The Clifford gates in the device, which we take to be the generators of the single-qubit Clifford group, are $H$ and $S$.  
We assume that $H$ can be implemented exactly but that $S$ has an over-rotation error such that
\begin{equation}
    S(\theta)= e^{-i \theta Z} S,
\end{equation}
for some value of $\theta$.  While this is called an ``over-rotation'' we make no assumption that $\theta>0$. 
We further apply depolarizing noise at a per-gate level to the system with strength $0.005$ meaning that we apply the channels
\begin{align}
    \Lambda_H &: \rho \mapsto 0.995 H \rho H + 0.005 (\openone/2),\nonumber\\
    \Lambda_{S(\theta)} &: \rho \mapsto 0.995 e^{-i \theta Z} S \rho S^\dagger e^{i\theta Z} + 0.005 (\openone/2).\label{eq:channel}
\end{align}

We assume that the user has control over the parameter $\theta$ but we do not assume that they know the functional form and thus do not know that setting $\theta=0$ will yield optimal performance.
The goal of our Bayesian ACRONYM algorithm is then to allow the method to discover that $\theta=0$ yields the optimal performance via local search.

\begin{figure}[t!]
    \begin{center}
        \includegraphics[width=0.8\linewidth]{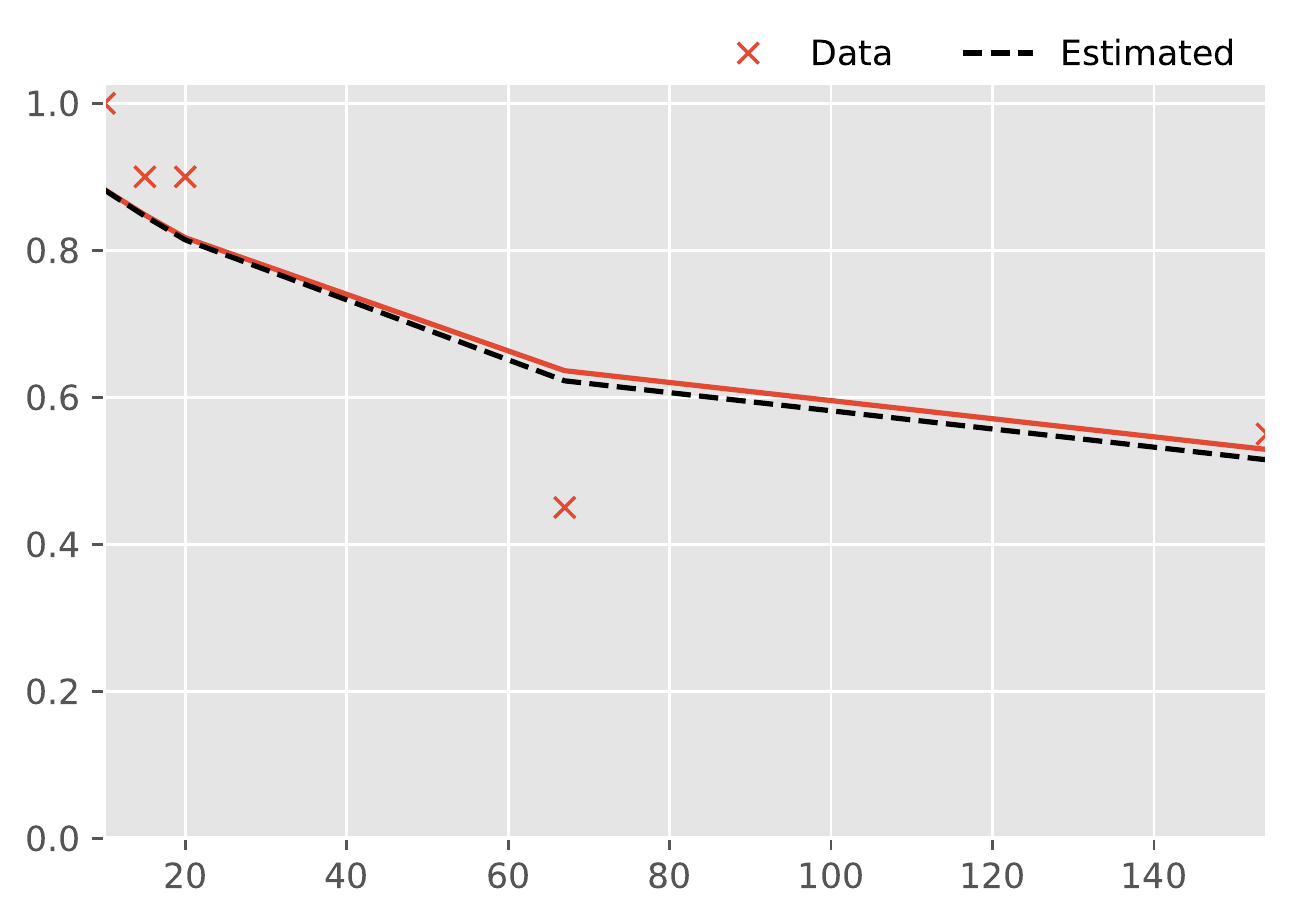}
    \end{center}
    \caption{
        \label{fig:unitary-overrotation-perf}
       Observed survival probabilities as a function of sequence lengths using $20$ measurements (shots) per length for an overrotation model with $\theta=0.04$.  
       Solid orange line represents the true value for the survival probability,  $(A-B)p^\LConst +B$, as a function of the sequence length $\LConst$ and the dashed line represents the estimate of the survival probability.  
       The prior was set to be uniform for $p$ and $A$ on $[0,1]$ and the prior $B$ was set to be the normal distribution~$\mathcal{N}(0.5,0.05^2)$.
    }
\end{figure}

\autoref{fig:unitary-overrotation-perf} shows the impact that using Bayesian inference to estimate RB parameters can have in data limited cases of the over-rotation problem.  Specifically, we apply Bayesian ACRONYM training to calibrate the over--rotation to within an error of $0.005$ which is equal to the dephasing error that we included in the channels in~\autoref{eq:channel}.
A broad prior was taken and despite the challenges that we would have learning a good model from least-squares fitting, we are able to accurately learn the survival probability.
 We can then learn the parameters $A$, $B$ and $p$, the latter of which gives us the average gate fidelity needed for ACRONYM training via~\autoref{eq:pDef}.
 As the required accuracy for the estimate of $p$ increases, the advantages gleaned from using Bayesian methods relative to fitting disappear~\cite{hwf+_bayesian_2018}. 
 However, in our context this observation is significant because we wish to tune the performance of quantum devices in the small data limit rather than the large data limit and use prior information from previous experiments to compensate.
 
 Local search is implemented using SPSA with learning rate $0.05$, a step of $0.05$ used to compute approximate gradients and a maximum step size of $0.1$.
We repeat the method until the posterior variance in the average gate fidelity is less than $0.005^2$.
We use a Lipshitz constant of $1.48$, which was numerically computed as a bound to give an appropriate amount of diffusion for the posterior distribution during an update.  
Bayesian inference is approximated using a particle filter with $256~000$ particles and Liu--West resampling with a resample threshold of $1/256$ as implemented by QInfer~\cite{granade2017qinfer}.
Single shot experiments are used with a maximum number of sequences of $500$ per set of parameters.

Perhaps the key observation is that throughout the tuning process the true parameters for the overrotation error remain within the $70\%$ credible region reported by QInfer, which suggests if anything that the credible region is pessimistic.
The estimate of $\happy$ also closely tracks the true throughout the learning process and also the amount of data required for the tuning process is minimal, less than $1$ kB.

\begin{figure}[t!]
    \begin{center}
        \includegraphics[width=\linewidth]{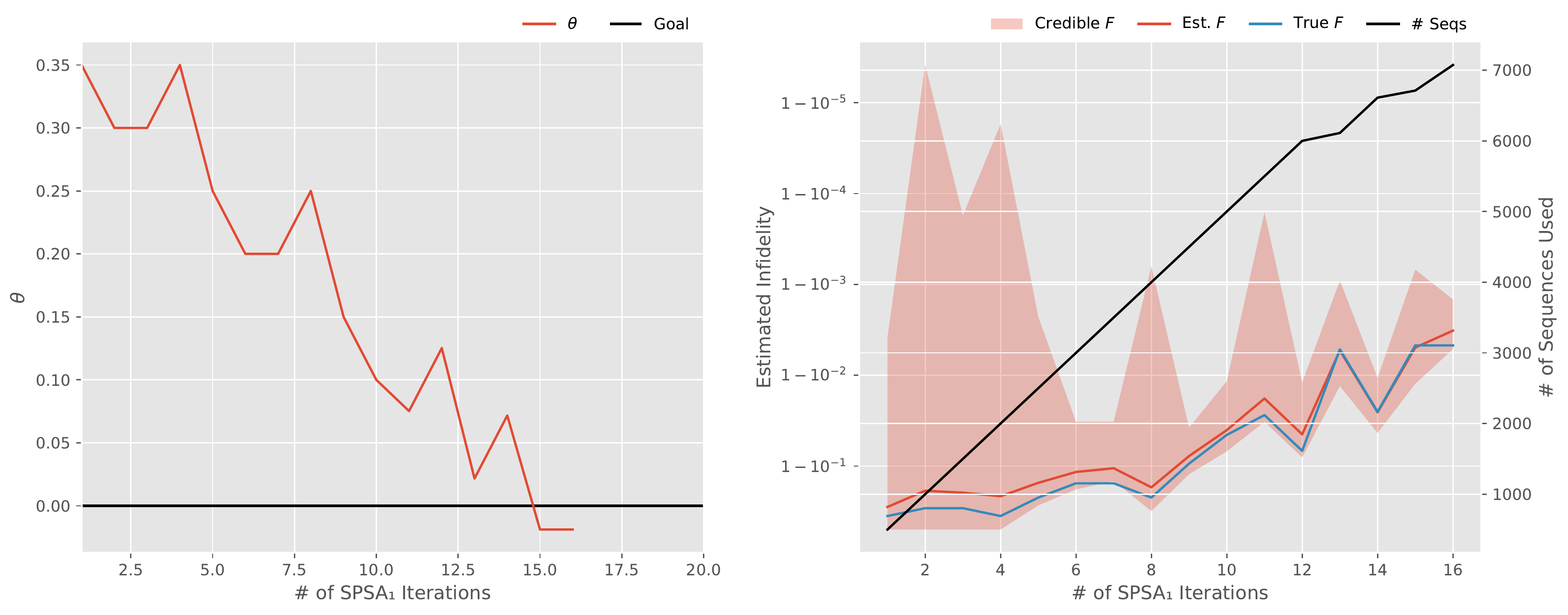}
    \end{center}
    \caption{
        \label{fig:unitary-overrotation-perf}
       Over-rotation angle and objective function values for an over-rotation model with a $0.35$ radian over-rotation initially with a target error of $0.005$ in $F$ as measured by the posterior standard-deviation.  (Left) Over-rotation angle as a function of number of iterations of SPSA taken.  (Right) Estimated Average gate infidelity as a function of the number of SPSA iterations and the total number of sequences used to achieve that level of infidelity.  The shaded region represents a $70\%$ credible region for the infidelity.
    }
\end{figure}

%=============================================================================
\section{Conclusion}
%=============================================================================
The main result of our work is to show that, under weak assumptions of Lipshitz continuity, Bayesian inference can be used to piece together evidence gained from experiments at nearby experimental settings to accelerate learning of optimal control parameters for quantum devices.  
We further demonstrate the success of this approach numerically by using a Bayesian ACRONYM tuning protocol (BACRONYM) to tune a rotation gate that suffers from an unknown overrotation.
We find that by use of evidence from nearby experimental settings for the gate, we can learn optimal controls with fewer than $1$ kilobit of data which is a reduction of nearly a factor of $20$ relative to the best known non-Bayesian approach~\cite{kelly2014optimal}.

Looking forward, there are a number of ways in which this work can be built upon.
Firstly, upper bounds on the Lipshitz constant and variance are needed to properly use evidence from nearby points within the optimization loop; however, tight estimates are not known a priori for either quantity.
Finding approaches that yield useful empirical bounds would be an important contribution beyond what we provide here.
Secondly, an experimental demonstration of Bayesian ACRONYM tuning would be useful to demonstrate the viability of such tuning parameters in real-world applications.
Finally, while we have picked SPSA as an optimizer for convenience, there may be better choices within the literature.
This raises an interesting issue because the number of times that the objective function needs to be queried is not the best metric when information is reused.
This point is important not just for choosing the best optimizer to minimize experimental costs for tuning hardware, it also potentially reveals a new way of optimizing parameters in variational quantum eigensolvers~\cite{peruzzo2014variational}, as well as QAOA~\cite{farhi2014quantum} and quantum machine learning algorithms~\cite{schuld2018circuit}.

\begin{acknowledgments}
    This project was prepared using a reproducible workflow \cite{granade_reproducible_2017}.
\end{acknowledgments}

\nocite{apsrev41Control}
\bibliographystyle{apsrev4-1}
\bibliography{apsrev-control,bayesian-acronym}

\appendix

%=============================================================================
\section{Pseudocode for BACROYNM Tuning}
%=============================================================================

\begin{algorithm}[H]
    \caption{\label{alg:bacronym}
        Bayesian ACRONYM tuning procedure
    }
    \begin{algorithmic}        
        \Function{BACRONYM}{}
            \Arguments
                \State $\vec{\theta}_0$: initial control parameters
                \State $n_{\text{shots}}$: number of measurements per seq. length
                \State $\sigma_{\text{req}}$: required accuracy for $\happy$
                \State $(a, b, s, t)$: SPSA1 parameters
                \State largest allowed step in the parameter $\vec{\theta}$
                \State $\happy_{\text{target}}$: target objective function value
                \State $\pi_0$: initial prior
                \State $\LConst$: Lipschitz continuity assumed for $\happy$
            \EndArguments
            \seccomment{Initialization}
            \State $\pi \gets \pi_0$, $\vec{\theta} \gets \vec{\theta}_0$
            \State collect RB data at $\vec{\theta}$ until $\Var[\happy] \le \sigma_{\text{req}}^2$
            \State $\hat{\happy} \gets \expect[\happy(\vec{\theta}) | \text{data}]$
            \State $i_{\text{iter}} \gets 0$
            \seccomment{Main body}
            \While{$\hat{\happy} \le \happy_{\text{target}}$}
                \State $i_{\text{iter}}+\!\!+$
                \seccomment{SPSA1}
                \State $\vec{\Delta} \gets$ a random $\pm1$ vector the same length as $\vec{\theta}$
                \State $\mathrm{step} \gets a / (1 + i_{\text{iter}}^{s})$
                \State $\mathrm{gain} \gets b / (1 + i_{\text{iter}}^{t})$
                \State $\vec{\delta\theta} \gets \mathrm{step} \cdot \vec{\Delta}$
                \State estimate $\hat{\happy}(\vec{\theta} + \vec{\delta\theta})$ using \autoref{cor:prior-reuse}
                \State $\vec{u} \gets \mathrm{gain} \cdot \vec{\Delta} (\hat{\happy}(\vec{\theta} + \vec{\delta\theta}) - \hat{\happy}(\vec{\theta}))$
                \If{any component of $\vec{u}$ larger than max update}
                    \State $\vec{u} \gets \vec{u} / \max_{u \in \vec{u}} |u|$
                \EndIf
                \If{$|\hat{\happy}(\vec{\theta} + \vec{\delta\theta}) - \hat{\happy}(\vec{\theta})| \ge \Var[\happy(\vec{\theta} + \vec{\delta\theta})]$}
                    \State $\vec{\theta} +\!\!= \vec{u}$
                    \inlinecomment{Complete the SPSA step.}
                \ElsIf{$\hat{\happy}(\vec{\theta} + \vec{\delta\theta}) < \hat{\happy}(\vec{\theta})$}
                    \State $\vec{\theta} -\!\!= \mathrm{step} \cdot \vec{\Delta}$
                \Else
                    \State $\vec{\theta} +\!\!= \mathrm{step} \cdot \vec{\Delta}$
                \EndIf
            \EndWhile
            \seccomment{Final estimate}
            \State \Return $\vec{\theta}, \hat{\happy}$
        \EndFunction
    \end{algorithmic}
\end{algorithm}
\end{document}